\documentclass[a4paper, reqno]{amsart}

\usepackage{amsmath, amssymb, amsthm}
\usepackage{mathrsfs}
\usepackage{mathtools}
\usepackage{dsfont}
\usepackage{array}
\usepackage{booktabs}
\usepackage{multirow}
\usepackage{microtype}
\usepackage{subcaption}
\usepackage{pgfplots}
\usepackage{listings}
\usepackage{algorithm}%
\usepackage{algorithmicx}%
\usepackage{algpseudocode}
\usepackage{url}
\usepackage[giveninits, maxbibnames=99, backend=biber, eprint=false, doi=false, isbn=false]{biblatex}
\usepackage[hyperfootnotes=false]{hyperref}

\usetikzlibrary{positioning, calc, shapes, backgrounds, decorations}
\usetikzlibrary{external}

\tikzset{
  contour/.style={
    decoration={
      name=contour lineto closed,
      contour distance=#1
    },
    decorate}}

\graphicspath{{figures/}}

\pgfplotsset{compat=newest}
\usetikzlibrary{arrows.meta}

\newcommand\tikzgraphsettings{\tikzset{
	every node/.style = {
		circle, fill,
		inner sep=0pt,
		minimum size=1.7mm,
		outer sep=1pt,
		auto
	},
	every label/.append style={font=\footnotesize},
	every mark/.append style={
		mark size=2.25pt
	},
	empty/.style = {
		font = \small,
		rectangle,
		inner sep=0pt,
		draw=none,
		fill=none
	},
	>=latex}
}

\definecolor{mplblue}{HTML}{1f77b4}
\definecolor{mplorange}{HTML}{ff7f0e}
\definecolor{mplgreen}{HTML}{2ca02c}
\definecolor{mplred}{HTML}{d62728}
\definecolor{mplpurple}{HTML}{9467bd}
\definecolor{mplbrown}{HTML}{8c564b}
\definecolor{mplpink}{HTML}{e377c2}
\definecolor{mplgray}{HTML}{7f7f7f}
\definecolor{mplyellowgreen}{HTML}{bcbd22}
\definecolor{mplcyan}{HTML}{17becf}

\theoremstyle{plain}
\newtheorem{theorem}{Theorem}
\newtheorem*{theorem*}{Theorem}
\newtheorem{proposition}{Proposition}

\theoremstyle{definition}
\newtheorem{definition}{Definition}

\theoremstyle{remark}
\newtheorem{remark}{Remark}

\DeclareMathOperator{\matsum}{sum}
\DeclareMathOperator{\diag}{diag}

\newcommand{\de}{\ensuremath{\mathrm{d}}}

\newcommand\1{\mathds{1}}

\newcommand\restr[2]{{\left.\kern-\nulldelimiterspace#1\mathchoice{\vphantom{\big|}}{}{}{}\right|_{#2}}}

\title{Optimization of geometric hypergraph embedding}

\author{Francesco Zigliotto}

\address{Scuola Normale Superiore. Pisa, PI 56126, Italy.}
\email{francesco.zigliotto@sns.it}

\author{Desmond J. Higham}
\address{School of Mathematics and Maxwell Institute,
University of Edinburgh, EH9 3FD, Scotland, U.K.}
\email{d.j.higham@ed.ac.uk}

\keywords{Euclidean embedding, hypergraph representation, link prediction, random geometric hypergraphs}
\subjclass{%
    05C65, 
    05C62, 
    68R10
}

\begin{document}

\def\djh#1{#1}
\def\fz#1{#1}

\begin{abstract}
We consider the problem of embedding the nodes of a hypergraph into Euclidean space under the assumption that the interactions arose through closeness to unknown hyperedge centres. In this way, we tackle the inverse problem associated with the generation of geometric random hypergraphs. We propose two new spectral algorithms; both of these exploit the connection between hypergraphs and bipartite graphs. The assumption of an underlying geometric structure allows us to define a concrete measure of success that can be used to optimize the embedding via gradient descent. Synthetic tests show that this approach accurately reveals geometric structure that is planted in the data, and tests on real hypergraphs show that the approach is also useful for the downstream tasks of detecting spurious or missing data and node clustering.
\end{abstract}

\maketitle
\thispagestyle{empty}

\maketitle

\begin{center}
\textbf{Dedicated to the memory of Nicholas J. Higham}
\end{center}

\bigskip

\section{Introduction}

Modelling pairwise interactions using a graph is a valuable technique in many application areas. However, there are numerous scenarios where relationships may involve more than two players. In terms of human interactions, we form groups to socialize, play sport, take public transport, create music, coauthor manuscripts, form committees, and so on. In other settings we may record collections of interacting proteins, reacting chemicals, co-purchased retail products or co-tagged images. In these cases, reducing the data to pairwise interactions, or \emph{flattening}, may involve an unnecessary loss of information, and hence it is attractive to represent the complete picture using a hypergraph \cite{Batt21,BCILLPYP21,benson2016higher,Bick23,torres2020why}. The hypergraph framework has proved useful for modelling the spread of disease and opinions \cite{ABAMPL21,EZG24,higham2021epidemics,heterogeneityHypergraph}, as well as in the study of traffic flow \cite{Yi20}, in the development of document recommendation systems~\cite{Zhu16}, and in semiconductor manufacturing \cite{Fuk84,Xia22}.

From an applied mathematics perspective, it is of interest to ask how interactions arise, and whether the underlying mechanisms can be inferred from the interaction data. In \cite{Bart22,HypergraphEmbeddingNature,RandomGeometricHypergraph2025}, extensions to the classic geometric graph \cite{Gil61} setting were proposed. Here, it is assumed that each node has a set of features that may be mapped into Euclidean space in such a way that nearby nodes interact. Distance may be interpreted literally (neighbours tend to shop at the same supermarkets and catch the same trains) or may be a more indirect measure (people with common interests may be more likely to socialize). 

We consider here the inverse problem---given the group-level connectivity structure, can we discover a suitable, geometrically consistent, embedding into Euclidean space? As well as revealing useful information about the data and aiding visualization \cite{Fuk84}, tackling this problem provides a first step  towards downstream clustering and prediction problems, including missing/spurious edge detection \cite{HypergraphEmbeddingNature}.

Hypergraph embedding is related to the more general  problem of hypergraph representation learning. Here \cite[Definition~4.1]{Ant23} we seek a ``latent representation of each node, which captures certain network topological information.'' Spectral methods are among the methods of choice in this area \cite{Ant23}, with various types of Laplacian being considered. In contrast with previous work, our approach, which assumes that the connectivity arose from an underlying geometric model involving (unknown) hyperedge centres, provides a natural quantitative measure of success. We show how this can feed into a descent algorithm that iteratively improves an embedding.

The paper is organized as follows: in Section~\ref{s:preliminaries}, we introduce some notation and describe the well-known equivalence between hypergraphs and bipartite graphs, which was also exploited from a hypergraph modelling context \cite{Bart22}. We also summarize the idea of spectral embedding via a graph Laplacian. In Section~\ref{s:embedding} we describe our new hypergraph embedding approaches, providing two related algorithms, which we test in Section~\ref{s:application}, along with their application to hyperedge error correction and clustering. Finally, in Section~\ref{s:conclusion}, we summarize this work and suggest directions for future work and possible generalizations.

\section{Preliminaries}
\label{s:preliminaries}

\subsection{Graphs and hypergraphs}

\begin{definition}
An (undirected\djh{, unweighted}) \emph{graph} is a pair $G=(V,E)$, where $V$ is a set of nodes and $E$ is a set of edges; each edge consists of an unordered pair $\{u,v\}$ of distinct nodes in $V$. 
\end{definition}


Every graph with ordered nodes $V=\{u_1,\dots,u_n\}$ can be associated with a
symmetric \fz{binary \emph{adjacency matrix} $A\in\mathbb{R}^{n\times n}$, such that $A_{ij}=1$ if and only if $\{u_i,u_j\}\in E$, and with the \emph{Laplacian matrix}}
\begin{equation}
\label{e:laplacian_def}
L=\diag(A\1)-A,
\end{equation}
where $\1$ is the vector of all ones.

In this work we employ a particular class of graphs, called \emph{bipartite graphs}.
\begin{definition}
A graph $G=(V,E)$ is \emph{bipartite} if the set of nodes $V$ can be partitioned into two non-empty subsets $V_1$ and $V_2$ such that every edge in $E$ connects a node in~$V_1$ to one in $V_2$.
\end{definition}

Graphs can be generalized by extending the concept of pairwise node interactions to interactions among arbitrary groups of nodes.
\begin{definition}
A \emph{hypergraph} is a pair $H=(V,E)$, where $V$ and $E$ are the sets of nodes and hyperedges. Each hyperedge $h\in E$ is of the form $h=\{u_1,\dots,u_k\}$ and can be any subset of $V$.
\end{definition}

A hypergraph with nodes $u_1,\dots,u_n$ and hyperedges $h_1,\dots,h_s$ can be represented via the \emph{incidence matrix} $B\in \mathbb{R}^{n\times s}$, where
\[
B_{ij}=\begin{cases}
1 & \text{if $u_i\in h_j$}\\
0 & \text{otherwise.}
\end{cases}
\]

There is a natural correspondence between hypergraphs and bipartite graphs.
\begin{definition}
\label{d:correspondence}
Given a hypergraph $H$ with nodes $u_1,\dots,u_n$ and hyperedges $h_1,\dots,h_s$, we define the \emph{incidence graph of $H$} as the
bipartite graph $G$ with nodes $\{u_1,\dots,u_n,h_1,\dots,h_s\}$ such that $\{u_i,h_j\}$ is an edge of $G$ if and only if $u_i$ belongs to $h_j$ in $H$. This correspondence also works in the reverse direction.
\end{definition}
\begin{remark}
Let $B$ be the incidence matrix of a hypergraph $H$ and let $G$ be the incidence graph of $H$. Then, the adjacency matrix of $G$ is
\begin{equation}
A = \begin{bmatrix}
    0 & B\\
    B^\top & 0
\end{bmatrix}\!.
\label{eq:AB}
\end{equation}
See Fig.~\ref{f:bipartite_correspondence} for an illustration of this correspondence.
\end{remark}
\djh{%
In this work, we are interested in hypergraphs with unweighted hyperedges, which implies that the corresponding bipartite graphs are also unweighted. However, the algorithms
that we develop will perturb the binary
adjacency matrix at intermediate steps, for improved results.
 For this reason, we introduce the notion of \emph{weighted graphs}, where each edge is assigned a nonzero real number; negative weights are permitted.
}

\begin{definition}\label{def:wg}
A \emph{weighted graph} is a tuple $(V,E,\omega)$, where the edges are weighted by a function $\omega:E\to\mathbb{R}\setminus\{0\}$ (edges with null weight are treated as non-edges). For~$V=\{u_1,\dots,u_n\}$, the adjacency matrix of $G$ is
\[
A_{ij}=
\begin{cases}
\omega\bigl(\{u_i,u_j\}\bigr) & \text{if $\{u_i,u_j\}\in E$}\\
0 & \text{otherwise}
\end{cases}
\]
and the Laplacian matrix $L$ is defined as in~\eqref{e:laplacian_def}.
\end{definition}

\begin{figure}
\centering
\subfloat[Hypergraph $H$]{
    \centering
    \input{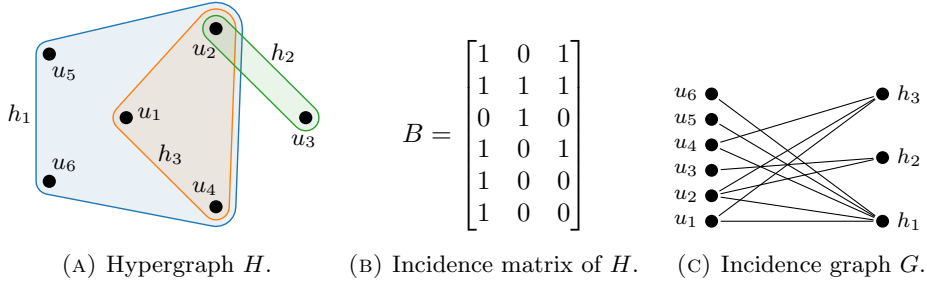}
}%
\subfloat[Incidence matrix of $H$]{
\raisebox{1.3cm}{\quad$\displaystyle
B=\begin{bmatrix}
  1 & 0 & 1\\
  1 & 1 & 1\\
  0 & 1 & 0\\
  1 & 0 & 1\\
  1 & 0 & 0\\
  1 & 0 & 0\\
\end{bmatrix}
$}\qquad\quad
}\hfill
\subfloat[Incidence graph $G$]{
    \begin{tikzpicture}
\tikzgraphsettings
\begin{scope}[scale=1.5]
\node (0) at (-0.5,-0.5625) [label=left:$u_1$] {};
\node (1) at (-0.5,-0.33749999999999997) [label=left:$u_2$] {};
\node (2) at (-0.5,-0.11249999999999999) [label=left:$u_3$] {};
\node (3) at (-0.5,0.11250000000000009) [label=left:$u_4$] {};
\node (4) at (-0.5,0.3375) [label=left:$u_5$] {};
\node (5) at (-0.5,0.5625) [label=left:$u_6$] {};
\node (6) at (1.0,-0.5625) [label=right:$h_1$] {};
\node (7) at (1.0,0) [label=right:$h_2$] {};
\node (8) at (1.0,0.5625) [label=right:$h_3$] {};
\draw (0)  to (6);
\draw (0)  to (8);
\draw (1)  to (6);
\draw (1)  to (7);
\draw (1)  to (8);
\draw (2)  to (7);
\draw (3)  to (6);
\draw (3)  to (8);
\draw (4)  to (6);
\draw (5)  to (6);
\end{scope}
\end{tikzpicture}
}
\caption{Example of a hypergraph along with its incidence matrix and incidence graph.}
\label{f:bipartite_correspondence}
\end{figure}

\subsection{Random geometric hypergraph model }
\label{s:rgh}

The connection between hypergraphs and bipartite graphs was leveraged in~\cite{RandomGeometricHypergraph2025} to introduce the random geometric hypergraph model (RGH). With this model, $n+s$ points $u_1,\dots,u_n$ and $h_1,\dots,h_s$ are sampled uniformly in a Euclidean domain of dimension $D$. Given a radius~$r$, the 
\djh{(unweighted}) 
bipartite graph $G$ is then constructed with nodes $u_1,\dots,u_n$ and $h_1,\dots,h_s$, such that $u_i$ is connected to $h_j$ whenever the geometric condition
\begin{equation}
\label{e:distance_condition}
\|u_i-h_j\|_2\le r
\end{equation}
is satisfied. The corresponding hypergraph $H$ can be built as outlined in Definition~\ref{d:correspondence}, so that $u_1,\dots,u_n$ are the nodes of $H$ and $h_1,\dots,h_n$ represent the hyperedges. We  refer to $h_1,\dots,h_n$ in the Euclidean domain as \emph{hyperedge centres}. Figure~\ref{fig:random_geometric} shows an example in $2$D space. This was used to generate the bipartite graph $G$ and hypergraph $H$ of Fig.~\ref{f:bipartite_correspondence}.

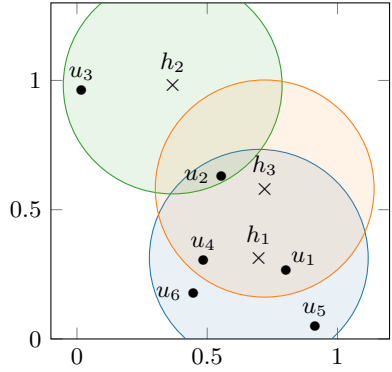
\begin{figure}
    \centering
    \begin{tikzpicture}[label distance=-5pt, every node/.style={font=\small}]
\begin{axis}[axis equal image, width=7cm, xmin=-0.1, xmax=1.2, ymin=0, ymax=1.3]
\draw[fill=mplblue, draw=mplblue, fill opacity=.1] (axis cs:0.6980, 0.3133) circle [radius=0.42];
\draw[fill=mplgreen, draw=mplgreen, fill opacity=.1] (axis cs:0.3675, 0.9810) circle [radius=0.42];
\draw[fill=mplorange, draw=mplorange, fill opacity=.1] (axis cs:0.7207, 0.5818) circle [radius=0.42];
\node at (0.8021, 0.2660) [label={[label distance=-10pt]above right:$u_1$}] {$\bullet$};
\node at (0.5538, 0.6283) [label=left:$u_2$] {$\bullet$};
\node at (0.0174, 0.9605) [label=$u_3$] {$\bullet$};
\node at (0.4850, 0.3017) [label=$u_4$] {$\bullet$};
\node at (0.9133, 0.0480) [label=$u_5$] {$\bullet$};
\node at (0.4468, 0.1759) [label=left:$u_6$] {$\bullet$};
\node at (0.6980, 0.3133) [label=$h_1$] {$\times$};
\node at (0.3675, 0.9810) [label=$h_2$] {$\times$};
\node at (0.7207, 0.5818) [label=$h_3$] {$\times$};
\end{axis}
\end{tikzpicture}
    \caption{Random set of $6$ nodes and $3$ hyperedges in $2$D space used to generate the hypergraph $H$ of Fig.~\ref{f:bipartite_correspondence}.}
    \label{fig:random_geometric}
\end{figure}

\subsection{Spectral embedding}
\label{s:spectral_embedding}

We conclude this section by presenting another tool that will be used throughout the paper. Given a graph $G$, the task of embedding the nodes $u_1,\dots,u_n$ into a Euclidean space in a way that best reflects the connectivity structure has been widely studied in the literature \cite{Belkin08,GHKstruggle24,spec_hkk,LuxburgUlrike2007Atos}. We may regard a $D$-dimensional embedding of $G$ as a matrix
\begin{equation}
Y=\begin{bmatrix}
Y(u_1)\\\vdots\\ Y(u_n)
\end{bmatrix}\in\mathbb{R}^{n\times D}.
\label{eq:Ydef}
\end{equation}
The $i$-th row $Y(u_i)$ corresponds to the embedding coordinates of the node $u_i$.

In this work, we make use of one of the most commonly adopted techniques, known as \emph{spectral embedding}.
\djh{Following on from the remarks that motivated 
Definition~\ref{def:wg}, it is convenient for us to formulate the problem in the general setting of a weighted graph $G$, possibly with negative weights.} \fz{The symmetric Laplacian $L$ can then be orthogonally diagonalized by the spectral theorem:
\begin{equation}
\label{e:l_ordering}
L=X\Lambda X^\top,
\end{equation}
where}
\djh{$\Lambda$ is a diagonal matrix having the real eigenvalues of $L$ on the diagonal, with 
the corresponding orthonormal eigenvectors
forming the columns of $X$.
By construction, the Laplacian $L$ has an eigenvector 
proportional to $\1$, with a corresponding eigenvalue of zero.
For convenience in the material that follows, we label this eigenvalue and eigenvector pair 
$\lambda_n$ and $v^{(n)}$ and then order the remaining 
eigenvalues so that 
$\lambda_1 \le \lambda_2\le \dots \le \lambda_{n-1}$,
with the corresponding eigenvectors denoted
$v^{(1)}, v^{(2)}, \ldots, v^{(n-1)}$.
Note that these remaining eigenvalues may take negative or zero values.
With this notation, we have 
 the following definition.}
\begin{definition}
\label{d:spectral_embedding}
For a weighted graph $G$ with Laplacian $L$, and for \fz{$D\le n-1$}, a \emph{spectral embedding with dimension $D$} of the nodes $u_1,\dots,u_n$ of $G$ is given by
\[
Y
=\left[\begin{array}{@{}c|c|c@{}}
v^{(1)} & \dots & v^{(D)}
\end{array}\right]\!,
\]
where the $i$-th row contains the coordinates of the embedding of $u_i$.
\end{definition}

The following standard proposition, 
see, for example,
\cite{Belkin08,GHKstruggle24}, shows that under certain assumptions, the spectral embedding is optimal: it may be interpreted as minimizing the distance between nodes connected by positive-weighted edges and maximizing the distance between those connected by negative-weighted edges.
\djh{We have not seen a proof in the literature that covers the case where weights may be negative. Hence, we 
include a brief proof here.}

\begin{proposition}
\label{p:spectral_embedding}
Let $G$ be a weighted graph with nodes $u_1,\dots,u_n$, and adjacency matrix~$A$. Then the spectral embedding of $G$ (Definition~\ref{d:spectral_embedding}) minimizes the quantity
\begin{equation}
\label{e:min_problem_spectral}
\sum_{i,j}A_{ij}\left\|Y(u_i)-Y(u_j)\right\|_2^2,
\end{equation}
over all $D$-dimensional embeddings $Y\in\mathbb{R}^{n\times D}$ \fz{such that $Y^\top  Y=I$ and $Y^\top \1=0$}.
\end{proposition}
\begin{remark}
The column normalization and sum conditions guarantee that the embedding is centred around $0$ and contained in $[-1,1]^D$.
\end{remark}

\fz{%
\begin{proof}\,
Let $Y$ be a $D$-dimensional embedding satisfying $Y^\top  Y=I$ and $Y^\top \1=0$. It is straightforward to verify that
\begin{equation}
\label{e:trace_equality}    
\dfrac1{2}\sum_{i,j}A_{ij}\left\|Y(u_i)-Y(u_j)\right\|_2^2=\operatorname{Tr}(Y^\top LY),
\end{equation}
where $\operatorname{Tr}$ denotes the trace. 
\djh{Recall that 
$\Lambda = \diag(\lambda_1, \lambda_2, \ldots, \lambda_n)$
in (\ref{e:l_ordering}), with $\lambda_n = 0$
and the corresponding last column of $X$ proportional to $\1$.}
Setting $Q=X^\top Y$, we have
\begin{equation}
\label{e:min_firsteq}
\operatorname{Tr}(Y^\top LY)=\operatorname{Tr}(Q^\top \Lambda Q)=\sum_{i=1}^{n}\lambda_i\|q_i\|_2^2,
\end{equation}
where $q_i$ is the $i$-th row of $Q$. We note that $\|q_i\|_2\le 1$, as $Q$ has orthonormal columns,
\djh{
and $q_n$ is the zero vector, since 
$Y^\top \1 = 0$.} 
Moreover,
\[
\sum_{i=1}^{n-1}\|q_i\|_2^2=\|Q\|^2_{\mathrm{F}}=\|Y\|^2_{\mathrm{F}}=D,
\]
\looseness-1
where $\|\cdot\|_{\mathrm{F}}$ is the Frobenius norm. We conclude that the minimum value of~\eqref{e:min_firsteq} is attained when $\|q_i\|_2=1$ for 
\djh{$i = 1,\ldots,D$}; this is precisely what occurs for the spectral embedding.
\end{proof}}

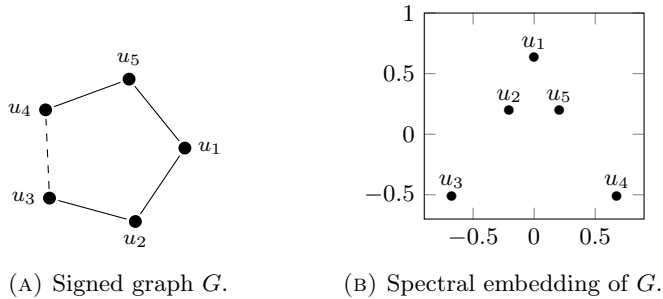
\begin{figure}
\centering
\subfloat[Signed graph $G$]{
    \begin{tikzpicture}
\tikzgraphsettings
\begin{scope}[scale=1]
\node (0) at (-0.807771184,-0.585253568) [label=left:\small$u_3$] {};
\node (1) at (-0.807771184, 0.586061834) [label=left:\small$u_4$] {};
\node (2) at ( 0.307268405, 0.944264921) [label={[label distance=-2pt]above:\small$u_5$}] {};
\node (3) at ( 1.000947240,-0.000326753) [label=right:\small$u_1$] {};
\node (4) at ( 0.307326724,-0.944746434) [label={[label distance=-2pt]below:\small$u_2$}] {};\draw (0)  to (4);
\draw (1)  to (2);
\draw (2)  to (3);
\draw (3)  to (4);
\draw[dashed] (0) to (1);
\end{scope}
\end{tikzpicture}
}
\qquad
\subfloat[Spectral embedding of $G$]{
    \begin{tikzpicture}[label distance=-6pt, every node/.style={font=\small}]
\begin{axis}[axis equal image, width=5cm, xmin=-0.9, xmax=.9, ymin=-0.7, ymax=1]
\node at (-0.0000, 0.6325) [label=\small$u_1$] {$\bullet$};
\node at (-0.2049, 0.1954) [label=\small$u_2$] {$\bullet$};
\node at (-0.6768, -0.5117) [label=\small$u_3$] {$\bullet$};
\node at (0.6768, -0.5117) [label=\small$u_4$] {$\bullet$};
\node at (0.2049, 0.1954) [label=\small$u_5$] {$\bullet$};
\end{axis}
\end{tikzpicture}
}
\caption{Example of a $2$D spectral embedding of a signed graph $G$, where the solid edges have weight $1$ and the dashed edge has weight $-1$.}
\label{f:small_signed_graph_embedding}
\end{figure}

\fz{Figure~\ref{f:small_signed_graph_embedding} shows an example of a small signed graph, with weights of $+1$ or $-1$, and the corresponding $2$D spectral embedding.}

Other variants of spectral embedding have been proposed for mixed-sign weighted graphs. In particular, the use of the signed Laplacian \[\bar L= \diag(|A|\1)-A,\] where $|A|$ denotes the entry-wise absolute value of $A$, has been suggested as an alternative to the standard Laplacian \cite{SignedLaplacian2010}. \fz{However, some works indicate that the standard Laplacian}
\djh{can be more effective \cite{SignedLF2017,SignedSpectralEmbedding2021}. 
In our application,
we use $L$ because Proposition~\ref{p:spectral_embedding} does not apply
to the signed Laplacian $\bar L$.}

\section{Hypergraph embedding}
\label{s:embedding}

We now formalize the main aim of this work---embedding a hypergraph into a given dimension $D$ under the assumption that the connectivity arose from a geometric model.

\subsection{Problem statement}
\label{s:problem}

Let $H^{(0)}$ be a given hypergraph, with nodes $u_1,\dots,u_n$ and hyperedges $h_1,\dots,h_s$, and let $G^{(0)}$ be its incidence graph, with $N=n+s$ nodes (see Definition~\ref{d:correspondence}). We are looking for both a $D$-dimensional embedding $Y$ (\ref{eq:Ydef}) of the nodes of $G^{(0)}$, and a radius $r$, such that $G^{(0)}$ can be recovered by applying the geometric hypergraph model to the points of the embedding (see Section~\ref{s:rgh}). Specifically, to assess the quality of an embedding $Y$, we construct the bipartite graph $\tilde G(Y,r)$ with $N = n + s$ nodes such that an edge connects $u_i$ to $h_j$ whenever $d_{ij}(Y)\le r$, where
\begin{equation}
\label{e:d_Y}
d_{ij}(Y)=\|Y(u_i)-Y(h_j)\|_2,
\qquad 1 \le i \le n, \quad 1 \le j \le s, 
\end{equation}
records the Euclidean distance between the embeddings of $u_i$ and $h_j$. The goal is to choose $Y$ and $r$ such that the recovered bipartite graph $\tilde G(Y,r)$ is equal to the original bipartite graph $G^{(0)}$. For simplicity, we will sometimes simply write~$\tilde G$ instead of $\tilde G(Y,r)$.

The problem may be impossible to solve exactly, so we seek approximate solutions. To quantify the goodness of a candidate embedding $Y$ and radius $r$, we define the loss function
\[
\mathcal{L}(Y,r)=\dfrac{\bigl\|\tilde B(Y,r)-B^{(0)}\bigr\|^2_{\mathrm{F}}}{\bigl\|B^{(0)}\bigr\|^2_{\mathrm{F}}},
\]
where $B^{(0)}$ is the incidence matrix of $H^{(0)}$ and
\begin{equation}
\label{e:step_function}
\tilde B_{ij}(Y,r)=
\begin{cases}
    1 & \text{if $d_{ij}(Y)\le r$}\\
    0 & \text{otherwise,}
\end{cases}
\end{equation}
is the incidence matrix of the hypergraph $\tilde H$ associated with $\tilde G$. 
Intuitively, $\mathcal{L}(Y,r)$ counts the relative number of mistakes in the reconstructed graph $\tilde G$; that is, the number of missing or spurious nodes in $\tilde H$'s hyperedges. We can thus formulate our problem as minimizing the function $\mathcal{L}(Y,r)$ over all possible $Y\in\mathbb{R}^{N\times D}$ and $r>0$.

\begin{remark}
\label{r:rgh_generated}
When $H^{(0)}$ is generated using the RGH model, choosing the original embedding yields an exact solution, since $B^{(0)}=\tilde B(Y,r)$ by definition.
\end{remark}

\subsection{Spectral embedding of the bipartite graph}
\label{subsec:standard}

A simple approach for the embedding problem is to set $Y$ as a spectral embedding (see Section~\ref{s:spectral_embedding}) of $A^{(0)}$, the adjacency matrix of the bipartite graph $G^{(0)}$. Proposition~\ref{p:spectral_embedding} shows that $u_i$ and $h_j$ will typically be closely placed whenever they are connected in $G^{(0)}$.
However, this method treats the nodes and 
the (artificially introduced) hyperedge centres
equally; whereas 
our ultimate goal is to reproduce the hyperedges, with the centres playing only an accompanying role.
Nevertheless, in some cases, particularly when $G^{(0)}$ is small, this direct approach may yield sufficiently accurate solutions. In Fig.~\ref{f:small_spectral_embedding}, we apply this technique to the graph $G$ from Fig.~\ref{f:bipartite_correspondence}. Since $G$ was generated by the RGH model (Fig.~\ref{fig:random_geometric}), we know that, in principle, it can be embedded and recovered exactly (see Remark~\ref{r:rgh_generated}). In this case, the spectral embedding of $G$, combined with an appropriate choice of $r$, results in perfect reconstruction.

\begin{figure}
\centering
\begin{tikzpicture}[label distance=-6pt, every node/.style={font=\small}]
\begin{axis}[axis equal image, width=7cm, xmin=-0.7, xmax=.8, ymin=-0.8, ymax=0.8]
\draw[fill=mplblue, draw=mplblue, fill opacity=0.1] (axis cs:-0.1875, 0.0413) circle [radius=0.5];
\draw[fill=mplgreen, draw=mplgreen, fill opacity=0.1] (axis cs:0.4652, 0.0151) circle [radius=0.5];
\draw[fill=mplorange, draw=mplorange, fill opacity=0.1] (axis cs:-0.1136, -0.4126) circle [radius=0.5];
\node at (0.0618, -0.1712) [label=$u_2$] {$\bullet$};
\node at (0.7088, 0.1875) [label=$u_3$] {$\bullet$};
\node at (-0.1818, -0.3436) [label={$u_1,u_4$}] {$\bullet$};
\node at (-0.2856, 0.5135) [label={$u_5,u_6$}] {$\bullet$};
\node at (-0.1875, 0.0413) [label=$h_1$] {$\times$};
\node at (0.4652, 0.0151) [label=$h_2$] {$\times$};
\node at (-0.1136, -0.4126) [label=right:$h_3$] {$\times$};
\end{axis}
\end{tikzpicture}
\caption{Spectral embedding of the nodes of $G$ from Fig.~\ref{f:bipartite_correspondence}. By choosing $r=0.5$ we recover the original graph exactly.}
\label{f:small_spectral_embedding}
\end{figure}
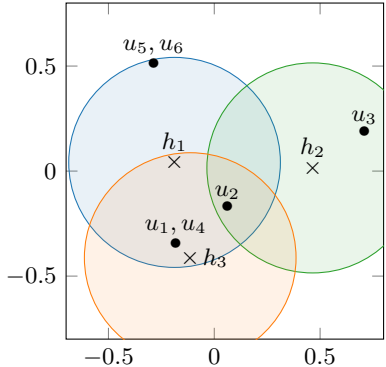

More generally, we found that the limited flexibility of this method may lead to poor results. In the following sections, we present two alternative approaches aimed at improving the performance of basic spectral embedding.

\subsection{Embedding optimization via weight modification}\label{subsec:GDSE}

The spectral embedding used in subsection~\ref{subsec:standard} offers useful properties, including the orthogonality constraint on the columns of $Y$ and a principled interpretation as the solution to the minimization problem~\eqref{e:min_problem_spectral}. The approach developed in this subsection retains these properties, but introduces greater flexibility by allowing \fz{weight modifications to the initially binary adjacency matrix} of $G^{(0)}$.

We begin by introducing some notation. Given a weighted bipartite graph $G$ with adjacency matrix~$A$ of the form (\ref{eq:AB}),
we let $Y^{\mathrm{sp}}(B)$
denote a spectral embedding (in dimension $D$) of $G$.

Our approach is now to minimize the loss function $\mathcal{L}(Y,r)$ with the extra constraint 
\[
Y = Y^{\mathrm{sp}}(B),
\]
so that the embedding~$Y$ must be the spectral embedding of the weighted bipartite graph associated with some matrix~$B\in\mathbb{R}^{n\times s}$. While the standard spectral embedding method described in 
subsection~\ref{subsec:standard}
fixes $B=B^{(0)}$ (the incidence matrix of $H^{(0)}$), here we instead allow for perturbations of $B^{(0)}$.
Intuitively, modifying the weights of $B^{(0)}$ to adjust the node embedding
allows us to search over a plausible, but constrained,
set of $Y$.
We will allow the weights to take arbitrary real values, including both positive and negative numbers. \fz{For best results, we allow perturbations on the \emph{full} incidence matrix~$B^{(0)}$, strengthening or weakening the contribution of both present and absent node-hyperedge relations.}

We will search for an optimal $B$ using gradient descent. To make the loss function  differentiable with respect to $B$, we modify the way that the reconstructed incidence matrix $\tilde B(Y,r)$ is built. Specifically, we approximate the discontinuous step function in~\eqref{e:step_function} with the differentiable function
\begin{equation}
\label{e:logfun}
f_\tau(x,r)=\dfrac{1}{1+e^{\tau^2(x^2-r^2)}},
\end{equation}
and we set
\begin{equation}
\label{e:approx_B}
[\tilde B_\tau]_{ij}(Y,r)=f_\tau\bigl(d_{ij}(Y),r\bigr).
\end{equation}
The corresponding reconstructed weighted bipartite graph $\tilde G_\tau$ is defined accordingly, where each~$u_i$ is connected to every~$h_j$ in $\tilde G_\tau$, with the weight of the edge  $\{u_i, h_j\}$ representing how close $u_i$ and $h_j$ are. As $\tau$ increases, $f_\tau$ better approximates the step function (see Fig.~\ref{f:logfun}) and the weights are encouraged towards $0$ and $1$.

\begin{figure}
    \centering
    \begin{tikzpicture}
\begin{axis}[
    xlabel={$x$},
    ylabel={$f_\tau(x, r)$},
    ylabel style = {rotate = -90},
    width=6.5cm,
    height=4cm,
    ymin=-.2,
    ymax=1.2,
    xmin=0,
    xmax=1,
    xticklabel style={
        font=\small,
    },
    yticklabel style={
        font=\small,
        /pgf/number format/fixed,
    },
    legend pos = outer north east,
]
\addplot [very thick, dashed, color=mplgreen] table {figures/logfun_5.txt};
\addplot [very thick, densely dashed, color=mplorange] table {figures/logfun_10.txt};
\addplot [very thick, color=mplblue] table {figures/logfun_20.txt};
\addlegendentry{$\tau=5$}
\addlegendentry{$\tau=10$}
\addlegendentry{$\tau=20$}
\end{axis}
\end{tikzpicture}
    \caption{The function $f_\tau(x, r)$
    in (\ref{e:logfun}) 
    with $r=0.5$ and different values of $\tau$. As $\tau$ increases, $f_\tau(x, r)$ becomes a better approximation of a step function centred at $r$.}
    \label{f:logfun}
\end{figure}
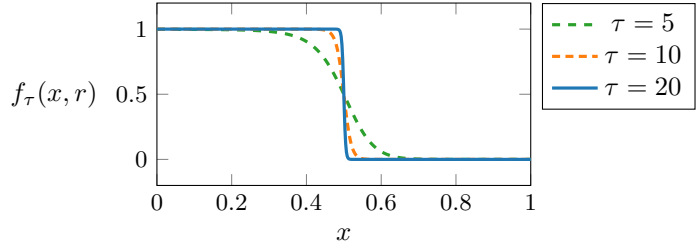

Under appropriate assumptions, the corresponding approximated loss function
\begin{equation}
\label{e:approx_loss}
\mathcal{L}_\tau(Y,r)=\dfrac{\bigl\|\tilde B_\tau\bigl(Y,r\bigr)-B^{(0)}\bigr\|^2_{\mathrm{F}}}{\bigl\|B^{(0)}\bigr\|^2_{\mathrm{F}}},
\end{equation}
with $Y=Y^{\mathrm{sp}}(B)$,
is differentiable with respect to $B$, $r$ and $\tau$, as stated by the following theorem.
\begin{theorem}
\label{t:gdse}
Let $B\in\mathbb{R}^{n\times s}$ be such that the Laplacian $L$ of the matrix
\[
A = \begin{bmatrix}
0 & B\\
B^\top & 0
\end{bmatrix}
\in \mathbb{R}^{N \times N} 
\] has distinct eigenvalues.  Let $Y=Y^{\mathrm{sp}}(B)$ be the $D$-dimensional spectral embedding of the weighted bipartite graph $G$ with  adjacency matrix $A$, so that the columns of $Y$ are $v^{(1)},\dots,v^{(D)}$. Let $d_{ij}$, $f_\tau$, $\tilde B_\tau$ and $\mathcal{L}_\tau$ be defined as in~\eqref{e:d_Y}, \eqref{e:logfun}, \eqref{e:approx_B} and~\eqref{e:approx_loss}.
Then $\mathcal{L}_\tau\bigl(Y^{\mathrm{sp}}(B),r\bigr)$ is differentiable with respect to $B$, $r$ and $\tau$, and we have
\begin{equation}
\label{e:dloss}
\begin{split}
\dfrac{\partial\mathcal{L}_\tau\bigl(Y^{\mathrm{sp}}(B),r\bigr)}{\partial B}
&=\dfrac{1}{\bigl\|B^{(0)}\bigr\|^2_{\mathrm{F}}}\sum_{k=1}^D\sum_{h=D+1}^{\fz{N-1}}\dfrac{\matsum\bigl(S\odot V^{(k)} \odot V^{(h)}\bigr)}{\lambda_k-\lambda_h}\:V^{(k)} \odot V^{(h)}\\
\dfrac{\partial\mathcal{L}_\tau\bigl(Y^{\mathrm{sp}}(B),r\bigr)}{\partial r}&=
\dfrac{1}{\bigl\|B^{(0)}\bigr\|^2_{\mathrm{F}}}\sum_{i=1}^n\sum_{j=1}^{s}2\left([\tilde B_\tau]_{ij}-B^{(0)}_{ij}\right)\dfrac{\partial f_\tau}{\partial r}(d_{ij}, r)\\
\dfrac{\partial\mathcal{L}_\tau\bigl(Y^{\mathrm{sp}}(B),r\bigr)}{\partial \tau}&=
\dfrac{1}{\bigl\|B^{(0)}\bigr\|^2_{\mathrm{F}}}\sum_{i=1}^n\sum_{j=1}^{s}2\left([\tilde B_\tau]_{ij}-B^{(0)}_{ij}\right)\dfrac{\partial f_\tau}{\partial \tau}(d_{ij}, r)
\end{split}
\end{equation}
where $V^{(k)}_{ij}=v_i^{(k)}-v_j^{(k)}$ and
\[
S_{ij} = \dfrac{2\Bigl([\tilde B_\tau]_{ij}-B_{ij}^{(0)}\Bigr) \dfrac{\partial f_\tau}{\partial d_{ij}}(d_{ij}, r)}{d_{ij}}.
\]
\end{theorem}
\begin{proof}\,
Recall that, for convenience, we are writing $\lambda_N = 0$ with eigenvector $v^{(N)}$ proportional to $\1$, and labelling the nonzero eigenvalues  (some of which may be negative) as $\lambda_1<\dots<\lambda_{N-1}$, with corresponding orthonormal eigenvectors $v^{(1)},\dots,v^{(N-1)}$.

From first-order perturbation theory \cite{Wilkinson1988}, and specifically \cite[Eq.~(3.11)]{Meyer1988}, we have
\[
\dfrac{\partial v_i^{(k)}}{\partial B_{IJ}}=\sum_{\substack{h=1\\[.5ex]h\ne k}}^{N}\dfrac{\left\langle\dfrac{\partial L}{\partial B_{IJ}}v^{(k)},\,v^{(h)}\right\rangle v_i^{(h)}}{\lambda_k-\lambda_h}
\]
where $L$ is the Laplacian matrix
\[
L = 
\begin{bmatrix}
\diag(B\1) & -B\\
-B^\top & \diag(B^\top\1)
\end{bmatrix}\!.
\]
Since
\[
\dfrac{\partial L}{\partial B_{IJ}}=e_Ie_I^\top +e_Je_J^\top -e_Ie_J^\top -e_Je_I^\top =(e_I-e_J)(e_I^\top -e_J^\top ),
\]
we obtain
\[
\dfrac{\partial v_i^{(k)}}{\partial B_{IJ}}=\sum_{\substack{h=1\\[.5ex]h\ne k}}^{N}\dfrac{\Bigl(v_I^{(k)}-v_J^{(k)}\Bigr)\Bigl(v_I^{(h)}-v_J^{(h)}\Bigr)\fz{v^{(h)}_i}}{\lambda_k-\lambda_h}=
\sum_{\substack{h=1\\[.5ex]h\ne k}}^{N-1}\dfrac{V^{(k)}_{IJ}V^{(h)}_{IJ}v_i^{(h)}}{\lambda_k-\lambda_h},
\]
where the term with $h=N$ vanishes, \fz{since $v^{(N)}$ has identical entries}. It follows that
\begin{equation}
\label{e:ded}
\begin{split}
\dfrac{\partial d_{ij}}{\partial B_{IJ}}
&=\dfrac{1}{2d_{ij}}\sum_{k=1}^D2\Bigl(v_i^{(k)}-v_j^{(k)}\Bigr)\sum_{\substack{h=1\\[.5ex]h\ne k}}^{N-1}\dfrac{V^{(k)}_{IJ}V^{(h)}_{IJ}\Bigl(v_i^{(h)}-v_j^{(h)}\Bigr)}{\lambda_k-\lambda_h}\\
&=\dfrac{1}{d_{ij}}\sum_{k=1}^D\sum_{\substack{h=1\\[.5ex]h\ne k}}^{N-1}\dfrac{V^{(k)}_{IJ}V^{(h)}_{IJ}V^{(k)}_{ij}V^{(h)}_{ij}}{\lambda_k-\lambda_h}
=\dfrac{1}{d_{ij}}\sum_{k=1}^D\sum_{h=D+1}^{N-1}\dfrac{V^{(k)}_{IJ}V^{(h)}_{IJ}V^{(k)}_{ij}V^{(h)}_{ij}}{\lambda_k-\lambda_h},  
\end{split}
\end{equation}
where the last equality holds because the double sum is antisymmetric in $k$ and $h$. Finally, we have
\[
\begin{split}
\dfrac{\partial\mathcal{L}_\tau\bigl(Y^{\mathrm{sp}}(B),r\bigr)}{\partial B_{IJ}}
&=\dfrac{1}{\bigl\|B^{(0)}\bigr\|^2_{\mathrm{F}}}\sum_{i=1}^{n}\sum_{j=1}^{s}2\left([\tilde B_\tau]_{ij}-B^{(0)}_{ij}\right)\dfrac{\partial f_\tau}{\partial d_{ij}}(d_{ij}, r)\dfrac{\partial d_{ij}}{\partial B_{IJ}}\\
&=\dfrac{1}{\bigl\|B^{(0)}\bigr\|^2_{\mathrm{F}}}\sum_{k=1}^D\sum_{h=D+1}^{\fz{N-1}}\left(\sum_{i=1}^{n}\sum_{j=1}^{s}S_{ij}V^{(k)}_{ij}V^{(h)}_{ij}\right)\dfrac{V^{(k)}_{IJ}V^{(h)}_{IJ}}{\lambda_k-\lambda_h}.
\end{split}
\]
The computation of the derivatives with respect to $r$ and $\tau$ is straightforward.
\end{proof}
\begin{remark}
If the condition about distinct eigenvalues is not met, \fz{it is sufficient} to apply a small random perturbation to $B$ and obtain distinct eigenvalues almost surely.
\end{remark}
\begin{remark}
This derivation provides some direct insight about how the spectral embedding changes when the weights of the matrix $B$ are perturbed. For instance, if the weight $B_{IJ}$ is increased by a small amount, then the distance between the embedding of~$u_I$ and~$h_J$ decreases, as expected. This follows directly from~\eqref{e:ded}:
\[
\dfrac{\partial d_{IJ}}{\partial B_{IJ}}
=\dfrac{1}{d_{ij}}\sum_{k=1}^D\sum_{h=D+1}^{\fz{N-1}}\dfrac{\Bigl(V^{(k)}_{IJ}V^{(h)}_{IJ}\Bigr)^2}{\underbrace{\lambda_k-\lambda_h}_{<0}}<0.
\]
\end{remark}

Theorem~\ref{t:gdse} shows how $B$, $r$, and $\tau$ can be updated iteratively via gradient descent (GD) on the loss function. The corresponding algorithm, which we refer to as \emph{Gradient-Descent Spectral Embedding} (GDSE), is summarized in Algorithm~\ref{a:gdse}.
\begin{algorithm}
\caption{Gradient-Descent Spectral Embedding (GDSE)}\label{a:gdse}
\begin{algorithmic}[1]
\Require hypergraph $H^{(0)}$ with incidence matrix $B^{(0)}$ and initial param. $r^{(0)}$, $\tau^{(0)}$
\State $B, r, \tau\gets B^{(0)}, r^{(0)}, \tau^{(0)}$
\State $Y \gets Y^{\mathrm{sp}}\left(B\right) $
\Repeat
    \State $\de B, \de r, \de\tau \gets \nabla_{B,r,\tau}\mathcal{L}_\tau\left(Y,r\right)$
    \State $B\gets B-\gamma_B\de B$
    \State $r\gets r-\gamma_r\de r$
    \State $\tau\gets \tau-\gamma_\tau\de\tau$
    \State $Y \gets Y^{\mathrm{sp}}\left(B\right)$
\Until{$\langle\text{\emph{stopping criterion}}\rangle$}\\
\Return $Y,\,r$
\end{algorithmic}
\end{algorithm}
The incidence matrix $B^{(0)}$ of $H^{(0)}$ is chosen as the starting point for $B$, while setting $r^{(0)}\approx0.1$ and $\tau^{(0)}\approx 5$ or $10$ appears to work well. The algorithm iteratively updates $B$, $r$ and $\tau$ via GD, using some fixed learning rates $\gamma_B, \gamma_r$, and $\gamma_\tau$.

An example run with $\gamma_B = 1$, $\gamma_r = 0.001$, and $\gamma_\tau = 1$ is shown in Fig.~\ref{f:reconstruction_art80_error_r_k}, where GDSE is applied to the random geometric hypergraph $H^{(0)}$ depicted in Fig.~\ref{f:hypergraph_art80}. 
The starting values here are $r^{(0)}=0.1$ and $\tau^{(0)}=5$. 
The embedding dimension is set to $D = 3$, matching the dimensionality used to generate $H^{(0)}$. After about $800$ iterations, the reconstruction error $\mathcal{L}$ stabilizes at $0$,
resulting in a perfect reconstruction of the original hypergraph. More examples are presented in Section~\ref{s:application}.

\begin{figure}
\centering
\subfloat[A random geometric hypergraph $H^{(0)}$, with $80$ nodes and $40$ hyperedges generated from a $3$-dimensional Euclidean space
(projected to 2D for visualisation).]{\label{f:hypergraph_art80}\raisebox{.5cm}{\includegraphics[width=.45\linewidth, angle=90]{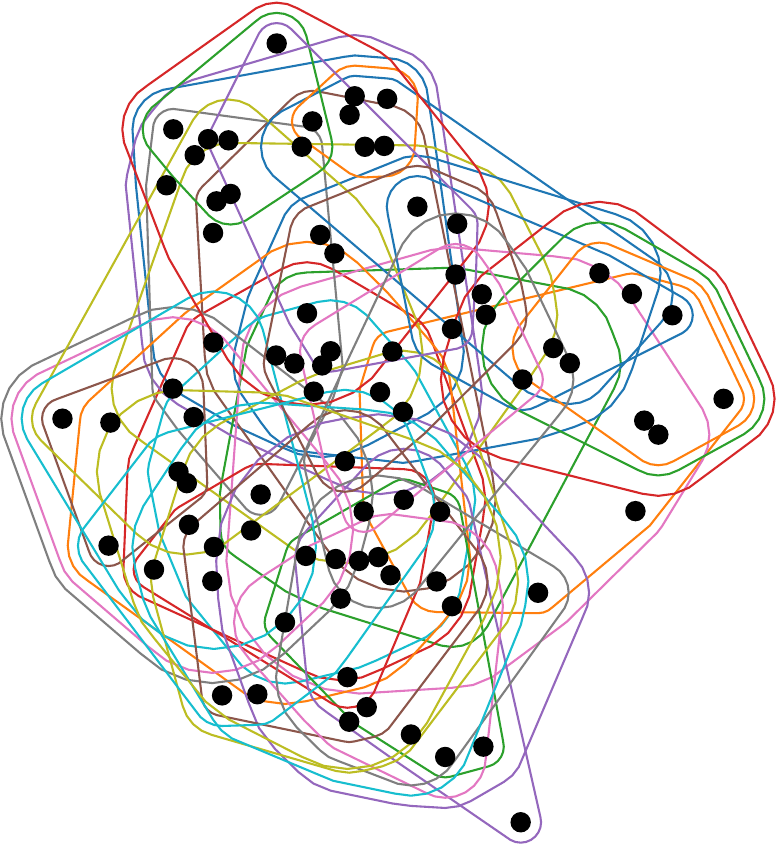}
}}\hfill
\subfloat[Progress of $\mathcal{L}_\tau(Y,r)$, $\mathcal{L}(Y,r)$, $r$ and $\tau$.]{\begin{minipage}[b]{.46\textwidth}
\raggedleft
\begin{tikzpicture}
\begin{semilogyaxis}[
    width=6.5cm,
    height=4.4cm,
    xticklabel style={
        font=\small,
    },
    xticklabels=\empty,
    yticklabel style={
        font=\small,
        /pgf/number format/fixed,
    },
    legend style={font=\footnotesize},
]
\addplot [very thick, color=mplblue] table [y expr=\thisrowno{0}, x expr=\coordindex] {figures/reconstruction_art80_GDSE_loss_list.txt};
\addplot [very thick, color=mplorange] table [y expr=\thisrowno{0}, x expr=\coordindex] {figures/reconstruction_art80_GDSE_error_list.txt};
\addlegendentry{$\mathcal{L}_\tau(Y,r)$}
\addlegendentry{$\mathcal{L}(Y,r)$}
\end{semilogyaxis}
\end{tikzpicture}\\
\begin{tikzpicture}
\begin{axis}[
    ylabel style = {rotate = -90},
    width=6.5cm,
    height=2.8cm,
    xticklabel style={
        font=\small,
    },
    xticklabels = \empty,
    yticklabel style={
        font=\small,
        /pgf/number format/fixed,
    },
    legend style={at={(.98,0.5)}, anchor=east, font=\footnotesize},
]
\addplot [very thick, color=mplgreen] table [y expr=\thisrowno{0}, x expr=\coordindex] {figures/reconstruction_art80_GDSE_r_list.txt};
\addlegendentry{$r$}
\end{axis}
\end{tikzpicture}\\
\begin{tikzpicture}
\begin{axis}[
    xlabel={iter.},
    ylabel style = {rotate = -90},
    width=6.5cm,
    height=2.8cm,
    xticklabel style={
        font=\small,
    },
    yticklabel style={
        font=\small,
        /pgf/number format/fixed,
    },
    legend style={at={(.98,0.5)}, anchor=east, font=\footnotesize},
]
\addplot [very thick, color=mplred] table [y expr=\thisrowno{0}, x expr=\coordindex] {figures/reconstruction_art80_GDSE_k_list.txt};
\addlegendentry{$\tau$}
\end{axis}
\end{tikzpicture}%
\end{minipage}%
\label{f:reconstruction_art80_error_r_k}}
\caption{Execution of the GDSE algorithm on $H^{(0)}$. Perfect reconstruction is achieved after around $800$ iterations.}
\label{f:reconstruction_art80}
\end{figure}
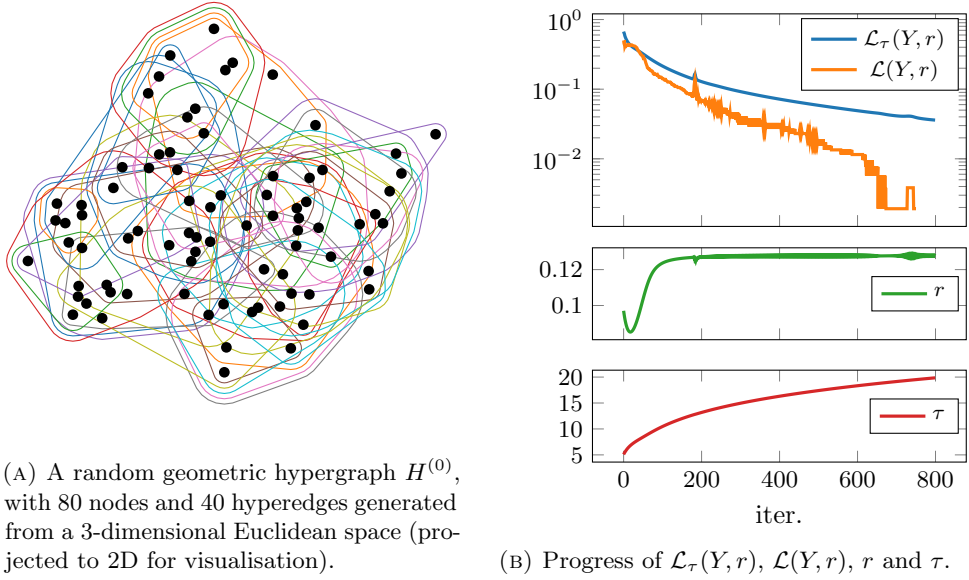

\begin{remark}
The GDSE algorithm tries to minimize the modified loss function $\mathcal L_\tau$, which is a good approximation of the original loss function $\mathcal{L}$ only for high values of~$\tau$. In earlier implementations of our algorithm, $\tau$ was hard-coded and increased at every iteration to ensure that it was high enough at the termination. However, 
\djh{when 
letting $\tau$ evolve automatically via GD, 
we observed that as $B$ and $r$ improved, the value of 
$\tau$ increased}  (Fig.~\ref{f:reconstruction_art80_error_r_k}).
Indeed, from~\eqref{e:dloss} we have
\[
\begin{split}
\dfrac{\partial\mathcal{L}_\tau}{\partial \tau}=\dfrac{1}{\bigl\|B^{(0)}\bigr\|^2_{\mathrm{F}}}\sum_{i=1}^n\sum_{j=1}^{s}2\left(f_\tau(d_{ij},r)-B^{(0)}_{ij}\right)(r-d_{ij})\,\underbrace{e^{\tau(d_{ij}-r)}f_\tau(d_{ij}, r)}_{>0}.
\end{split}
\]
The sign of each term of the double sum depends on the reconstruction quality of the embedding. Specifically, a non-negligible positive term arises whenever a node $i$, which does not belong to hyperedge $h$ in the original hypergraph, is nevertheless embedded closer than distance $r$ to $h$. As the embedding improves, such situations are expected to occur less frequently, so that ${\partial \mathcal{L}_\tau}/{\partial \tau}<0$, which in turn leads to an increase of $\tau$ in the next iteration.
\end{remark}

We note that several variations and extensions to Algorithm~\ref{a:gdse} could be implemented, including the use of momentum and decreasing learning rates.

\subsection{Direct optimization of the embedding}

The computation of $\partial L_\tau/\partial B$ in each GDSE iteration has time complexity $O(DNns)$, together with the cost of evaluating the full spectrum of $L \in \mathbb{R}^{N \times N}$,
which is intractable for large hypergraphs. In contrast, the approach presented in this subsection is more scalable. It imposes no constraint on $Y$, allowing for faster gradient descent updates, though at the cost of losing the desirable properties of a spectral embedding.

More precisely, the \emph{Gradient-Descent Embedding} (GDE) algorithm directly optimizes the embedding matrix $Y$ via gradient descent on the approximated loss $\mathcal{L}_\tau(Y,r)$ defined in~\eqref{e:approx_loss}. As in GDSE, both $r$ and $\tau$ are also optimized by the algorithm.
It is straightforward to verify that we have
\begin{equation}
\label{e:dloss_dY}
\begin{split}
\dfrac{\partial\mathcal L_\tau}{\partial Y(u_i)}
&=2\sum_{j=1}^{s}\left(\tilde B_{ij}-B^{(0)}_{ij}\right)\dfrac{\partial f_\tau}{d_{ij}}(d_{ij}, r)\dfrac{Y(u_i)-Y(h_j)}{d_{ij}},\quad\text{for $i=1,\dots,n$,}\\
\dfrac{\partial\mathcal L_\tau}{\partial Y(h_j)}&=2\sum_{i=1}^{n}\left(\tilde B_{ij}-B^{(0)}_{ij}\right)\dfrac{\partial f_\tau}{d_{ij}}(d_{ij}, r)\dfrac{Y(h_j)-Y(u_i)}{d_{ij}},\quad\text{for $j=1,\dots,s$,}
\end{split}
\end{equation}
while the derivatives with respect to $r$ and $\tau$ are the same as in Theorem~\ref{t:gdse}.

As for the choice of starting values, while the spectral embedding of $G^{(0)}$ 
typically provides a good initial guess for $Y$, a more efficient approach may be required for large hypergraphs, especially when $s \gg n$. In particular, we construct the undirected weighted graph $\mathcal{G}$ with nodes $u_1, \dots, u_n$ and adjacency matrix $B^{(0)}B^{(0)\top} \in \mathbb{R}^{n \times n}$, where the weight of the edge~$\{u, v\}$ corresponds to the number of hyperedges shared by $u$ and $v$. 
We use the spectral embedding of $\mathcal{G}$ to initialize the embeddings of the nodes $u_1, \dots, u_n$. Then, for each hyperedge $h_j$ we set $Y(h_j)$ as the centroid of the embeddings of the nodes it connects:
\[
Y(h_j) = \dfrac{1}{|h_j|}\sum_{u_i\in h_j} Y(u_i). 
\]
We denote the resulting initial embedding by  $Y^{\mathrm{c}}(\mathcal G)$.

The fast evaluation of the loss function in GDE, unlike GDSE that requires the computation of $Y^{\mathrm{sp}}(B)$, makes it feasible to
use more elaborate adaptive learning rates. In our implementation (Algorithm~\ref{a:gde}) we adopted the simple yet effective Armijo–Goldstein condition \cite{Armijo1966}. 
We also consider the corresponding 
\emph{stochastic gradient} approximation where the computation in~\eqref{e:dloss_dY} is restricted to fixed-size random subsets of nodes and hyperedges. 

\begin{algorithm}
\hfuzz2pt
\begin{algorithmic}
\caption{Gradient-Descent Embedding (GDE)}\label{a:gde}
\Require hypergraph $H^{(0)}$ with incidence matrix $B^{(0)}$ and initial parameters $r^{(0)}$ and $\tau^{(0)}$
\State $r, \tau\gets r^{(0)}, \tau^{(0)}$
\State $Y \gets Y^{\mathrm{sp}}\bigl(B^{(0)}\bigr)$ or $Y^{\mathrm{c}}(\mathcal G)$
\Repeat
    \State $\de Y, \de r, \de\tau \gets \nabla_{Y,r,\tau}\mathcal{L}_\tau\left(Y,r\right)$ (exact or stochastic)
    \State $\gamma_{\,Y}, \gamma_r, \gamma_\tau \gets \langle\text{\itshape step size selection strategy}\rangle \text{ (e.g., Armijo–Goldstein line search)}$
    \State $Y\gets Y-\gamma_{\,Y}\de Y$
    \State $r\gets r-\gamma_r\de r$
    \State $\tau\gets \tau-\gamma_\tau\de\tau$
\Until {$\langle\text{\emph{stopping criterion}}\rangle$}\\
\Return $Y,\,r$
\end{algorithmic}
\end{algorithm}

\section{Numerical experiments and application}
\label{s:application}

\subsection{Hypergraph reconstruction}
\label{s:hypergraph_reconstruction}

We now evaluate the performance of the GDSE and GDE algorithms on the hypergraph reconstruction task: given a hypergraph~$H^{(0)}$, we construct and optimize a Euclidean embedding of $H^{(0)}$ and compare the reconstructed geometric hypergraph $\tilde H$ with $H^{(0)}$, as described in Section~\ref{s:problem}.

The experiments were conducted on both synthetic and real-world hypergraphs. Specifically, \texttt{RGH1}, \texttt{RGH2}, and \texttt{RGH3} are geometric hypergraphs generated using the RGH model (Section~\ref{s:rgh}), with an increasing number of nodes and hyperedges. \texttt{RGH1} was generated from a 3-dimensional space, while the other two are based on a 6-dimensional space.
The \texttt{senate-committees} hypergraph represents members of the U.S. Senate as nodes, with hyperedges corresponding to committee memberships~\cite{chodrow2021hypergraph}.
The remaining two hypergraphs, \texttt{contact-primary-school} \cite{Stehl-2011-contact} and \texttt{contact-high-school} \cite{Mastrandrea-2015-contact}, encode student interactions, where each hyperedge represents a group of individuals who were in close proximity at a given time.

For the synthetic hypergraphs, we set the embedding dimension $D$ equal to the dimension of the space used to generate them. 
For the last two real-world hypergraphs, only GDE was tested, using the faster starting point $Y^{\mathrm{c}}(\mathcal G)$ and a stochastic  gradient computation, with a batch size of $256$. In all other cases, the full embedding $Y^{\mathrm{sp}}\bigl(B^{(0)}\bigr)$ was used and the full gradient was computed at every iteration.

The results are summarized in Table~\ref{t:reconstruction}. Both algorithms achieve a good reconstruction error for the smaller synthetic hypergraphs \texttt{RGH1} and \texttt{RGH2}, while GDSE struggles with larger graphs. In fact, GDE outperforms GDSE in most cases, in terms of both speed and reconstruction accuracy. The slightly higher time per iteration of GDE in the smaller hypergraphs is mostly due to the time-consuming Armijo-Goldstein line search procedure for learning rate selection, which improves the convergence rate.

\begin{table}
\caption{Performance of GDSE and GDE on both synthetic and real-world hypergraphs.
}\vspace{-1ex}
\label{t:reconstruction}
\footnotesize
\begin{tabular}{*7c}
\toprule
\normalsize hypergraph & \normalsize algorithm & \normalsize $D$ & \normalsize iter. & \normalsize time (s) & \normalsize time/it.\ (s) & \normalsize $\mathcal{L}$ \\
\midrule
\multirow{2}*{\begin{tabular}{c}\texttt{RGH1}\\ ($n=80$,  $s=40$)\end{tabular}} & GDSE & 3 & 800 & 26 & 0.03 & 0 \\
 & GDE & 3 & 150 & 6 & 0.04 & 0\\
\midrule
\multirow{2}*{\begin{tabular}{c}\texttt{RGH2}\\ ($n=160$,  $s=100$)\end{tabular}} & GDSE & 6 & 1000 & 184 & 0.18 & 0.033\\
& GDE & 6 & 200 & 40 & 0.20 & 0.037\\
\midrule
\multirow{2}*{\begin{tabular}{c}\texttt{RGH3}\\ ($n=300$,  $s=300$)\end{tabular}} & GDSE & 6 & 1000 & 1333 & 1.33 & 0.216\\
& GDE & 6 & 200 & 214 & 1.07 & 0.009\\
\midrule
\multirow{2}*{\begin{tabular}{c}\texttt{senate-committees}\\ ($n=282$,  $s=315$)\end{tabular}} & GDSE & 32 & 3000 & 9280 & 3.09 & 0.576 \\
 & GDE & 32 & 250 & 261 & 1.04 & 0.051 \\
\midrule
\begin{tabular}{c}\texttt{contact-high-school}\\ ($n=116$,  $s=2584$)\end{tabular} & GDE & 64 & 8000 & 9269 & 1.16 & 0.126 \\
\midrule
\begin{tabular}{c}\texttt{contact-primary-school}\\ ($n=242$,  $s=12704$)\end{tabular} & GDE & 64 & 8000 & 9242 & 1.16 & 0.227 \\
\bottomrule
\end{tabular}
\end{table}

\begin{figure}
\centering
\begin{tikzpicture}
\begin{semilogyaxis}[
    xlabel={iter.},
    ylabel={$\mathcal L_\tau$},
    ylabel style = {rotate = -90},
    width=6.7cm,
    height=4.6cm,
    xticklabel style={
        font=\small,
    },
    yticklabel style={
        font=\small,
        /pgf/number format/fixed,
    },
    legend style={font=\footnotesize},
    legend pos = outer north east,
]
\addplot [very thick, color=mplblue] table [y expr=\thisrowno{0}, x expr=\coordindex] {figures/reconstruction_sen4_GDE_armijo_loss_list.txt};
\addplot [very thick, color=mplorange] table [y expr=\thisrowno{0}, x expr=\coordindex] {figures/reconstruction_sen8_GDE_armijo_loss_list.txt};
\addplot [very thick, color=mplgreen] table [y expr=\thisrowno{0}, x expr=\coordindex] {figures/reconstruction_sen16_GDE_armijo_loss_list.txt};
\addplot [very thick, color=mplred] table [y expr=\thisrowno{0}, x expr=\coordindex] {figures/reconstruction_sen32_GDE_armijo_loss_list.txt};
\addlegendentry{$D=4$}
\addlegendentry{$D=8$}
\addlegendentry{$D=16$}
\addlegendentry{$D=32$}
\end{semilogyaxis}
\end{tikzpicture}
\vspace{-1ex}
\caption{Performance of GDE on \texttt{senate-committees} in different embedding dimensions.}
\label{f:reconstruction_sen_dims}
\end{figure}
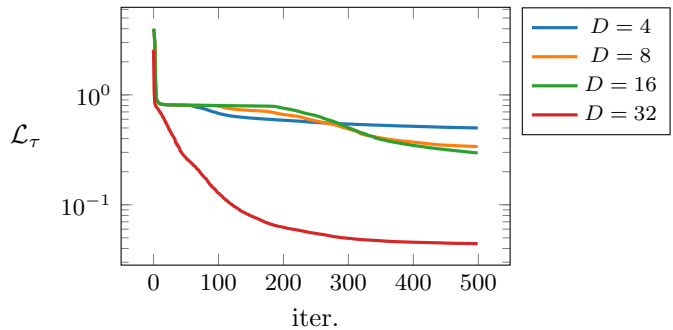

Given that the algorithms are effective at reproducing synthetic geometric hypergraphs, 
in a real-world setting we may view them as measuring the 
extent to which geometric structure is present in the data.
In Table~\ref{t:reconstruction} we see that 
GDE achieves a reconstruction error of around $5$\,\% for the \texttt{senate-committees} hypergraph 
and around 
$13$\,\% for the \texttt{contact-high-school} hypergraph,
while the value increases to around
$23$\,\% for the \texttt{contact-primary-school} hypergraph. Figure~\ref{f:reconstruction_sen_dims} shows that 
lower-dimensional embeddings yield significantly worse results.

In the following subsections, we present two possible applications of geometric hypergraph embedding. Although the ideas are applicable to both GDSE and GDE, we focus on the latter for simplicity.

\subsection{Detection of spurious or missing nodes in hyperedges}

Real-world hypergraph data may contain errors; notably spurious or missing nodes in hyperedges. Assuming geometric structure in the ground-truth hypergraph, geometric embeddings can be leveraged to identify such errors.

We first focus on spurious nodes. Let $H^+$ be a modified version of a hypergraph~$H$ that has the same nodes as~$H$, but with some ``enlarged'' hyperedges obtained by adding random nodes.
Given~$H^+$ we aim to recover the original hypergraph $H$.
In the geometric embedding of $H^+$, difficulty in satisfying the proximity condition between a node $u_i$ and a hyperedge $h_j$ may suggest that the corresponding connection in $H^+$ is spurious. Recall that the quantity $[\tilde B_\tau]_{ij}$, defined in~\eqref{e:approx_B}, reflects how close the embeddings of $u_i$ and $h_j$ are. If $u_i\in h_j$ in $H^+$ but $[\tilde B_\tau]_{ij}$ is below a certain threshold, $\alpha$, we may therefore classify $u_i$ as a spurious node of $h_j$.

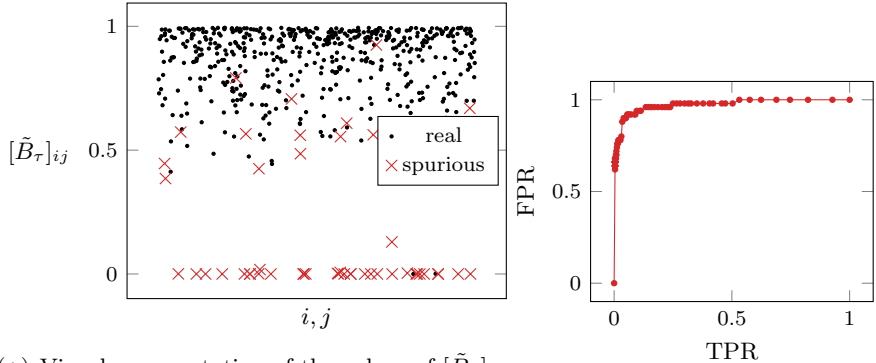
\begin{figure}
\centering
\subfloat[{Visual representation of the values of $[\tilde B_\tau]_{ij}$, for any $i,j$ such that $u_i\in h_j$ in $H^+$.}]{
\begin{minipage}[b]{.55\textwidth}
\begin{tikzpicture}
\begin{axis}[
    width=.98\textwidth, 
    height=5cm,
    xlabel={\small $i,j$},
    ylabel={\small $[\tilde B_\tau]_{ij}$},
    ylabel style = {rotate = -90},
    xtick=\empty,
    yticklabel style={
        font=\footnotesize,
        /pgf/number format/fixed,
    },
    legend style={at={(.98,0.5)}, anchor=east, font=\footnotesize}]
    
\addplot[black, only marks, mark=*, mark size=.6pt] table {figures/errcorr_art80_green.txt};
\addplot[mplred, only marks, mark=x, mark size=3pt] table {figures/errcorr_art80_blue.txt};
\addlegendentry{real}
\addlegendentry{spurious}
\end{axis}
\end{tikzpicture}%
\end{minipage}%
\label{f:errcorr_art_b_ij}}
\subfloat[ROC curve ($\mathrm{AUC} = 0.973$).]{
\begin{minipage}[b]{.43\textwidth}
\begin{tikzpicture}
\begin{axis}[
    width=5.3cm, 
    height=4.5cm,
    xlabel={\small TPR},
    ylabel={\small FPR},
    yticklabel style={font=\footnotesize},
    xticklabel style={font=\footnotesize},
    legend style={at={(0.4,1.1)}, anchor=south, font=\footnotesize}]
    
\addplot[mplred, mark=*, mark size=1pt] table {figures/errcorr_art80_added_roc.txt};
\end{axis}
\end{tikzpicture}
\end{minipage}%
\label{f:roc}}
\vspace{-1ex}
\caption{Adoption of geometric embedding for detection of $50$ spurious node-hyperedge relations in \texttt{RGH1}.}
\label{f:errcorr_art}
\end{figure}

Figure~\ref{f:errcorr_art} shows the results of this approach applied to a modified version of \texttt{RGH1}, in which $50$ spurious node-hyperedge relations were added at random---these are marked with a red cross in Fig.~\ref{f:errcorr_art_b_ij}. 
Many of these can be easily identified, as they correspond to entries where $[\tilde B_\tau]_{ij}\approx0$. 
For example, with $\alpha = 0.4$, $34$ out of $50$ spurious connections are correctly identified, while only two legitimate connections are misclassified. 

To judge the overall performance, we also show 
in Fig.~\ref{f:errcorr_art} the Receiver Operator Characteristic (ROC) curve, 
 with threshold parameter $\alpha$. The area under the curve (AUC) is $0.973$. 
Here, TPR is the true positive rate and 
FPR is the false positive rate, 
and an AUC of one indicates perfect reconstruction \cite{auc97}.

\begin{table}
\caption{Performance of GDE in detecting spurious or missing node-hyperedge relations in different hypergraphs.}\vspace{-1ex}
\label{t:missing-spurious}
\begin{tabular}{*7c}
\toprule
hypergraph & $D$ & iter. & AUC (spurious) & AUC (missing) \\
\midrule
\texttt{RGH1} & 3 & 50 & 0.973 & 0.977 \\
\texttt{RGH2} & 6 & 50 & 0.935 & 0.871 \\
\texttt{RGH3} & 6 & 50 & 0.958 & 0.959 \\
\texttt{senate-committees} & 32 & 100 & 0.912 & 0.891 \\
\bottomrule
\end{tabular}
\medskip
\end{table}

A similar analysis can be done for \emph{missing} nodes in hyperedges. In this case, we may search for high values of $[\tilde B_\tau]_{i,j}$ among the indices $i,j$ such that $u_i\not\in h_j$ in~$H^+$. Table~\ref{t:missing-spurious} summarizes the performance of GDE in identifying both spurious and missing node–hyperedge relations, when $50$ such relations are respectively added or removed.

\subsection{Community detection}

We now consider the use of geometric embedding for community detection in hypergraphs. We conducted our experiments on the student-interaction hypergraphs described in Section~\ref{s:hypergraph_reconstruction}, which include ground-truth community labels: $11$ for the former and $9$ for the latter. We applied the GDE algorithm, setting the embedding dimension to $16$. After each iteration, the nodes were clustered by applying K-means to their embeddings. Figure~\ref{f:clustering} shows the evolution of the Adjusted Rand Index (ARI) over the iterations.
The ARI score can range from $-1$ to $1$, with $1$ indicating perfect agreement and~$0$ being equivalent 
to random assignment \cite{ARI85}. Since K-means is highly sensitive to initialization, we report only the best result out of $50$ runs for each GDE iteration.

Although GDE requires thousands of iterations to achieve a low reconstruction error on the student-interaction hypergraphs (Table~\ref{t:reconstruction}), it turns out that a few dozen iterations are sufficient to significantly improve the clustering accuracy. The results compare favourably with  \cite{HypergraphEmbeddingNature}, where an ARI below $0.5$ is reported for both hypergraphs.

\begin{figure}
\centering
\subfloat[\texttt{contact-primary-school}]{\begin{tikzpicture}
\begin{axis}[
    xlabel={\small GDE iter.},
    ylabel={\small ARI},
    width=6.5cm,
    height=4cm,
    xticklabel style={
        font=\small,
    },
    yticklabel style={
        font=\small,
        /pgf/number format/fixed,
    },
]
\addplot [very thick, color=mplblue] table [y expr=\thisrowno{0}, x expr=\coordindex] {figures/clustering_primary_school_ari_list.txt};
\end{axis}
\end{tikzpicture}}
\subfloat[\texttt{contact-high-school}]{\begin{tikzpicture}
\begin{axis}[
    xlabel={\small GDE iter.},
    ylabel={\small ARI},
    width=6.5cm,
    height=4cm,
    xticklabel style={
        font=\small,
    },
    yticklabel style={
        font=\small,
        /pgf/number format/fixed,
    },
]
\addplot [very thick, color=mplblue] table [y expr=\thisrowno{0}, x expr=\coordindex] {figures/clustering_high_school_ari_list.txt};
\end{axis}
\end{tikzpicture}}
\caption{Evolution of ARI (best of $50$ K-means runs) over GDE iterations on student-interaction hypergraphs.}
\label{f:clustering}
\end{figure}
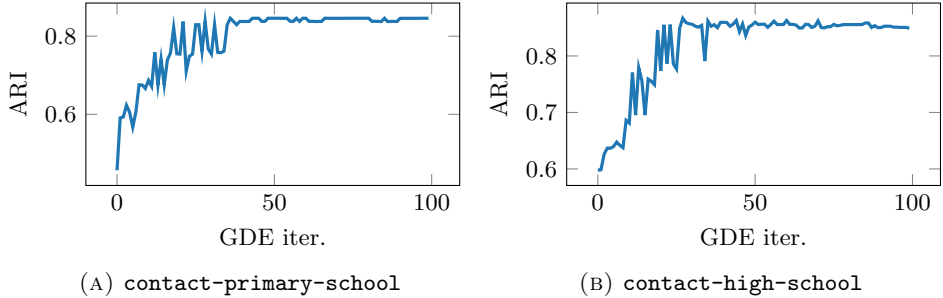

\section{Conclusions and future work}
\label{s:conclusion}

Spectral methods for node embedding have a long tradition in graph theory.
Here we developed new spectral methods that apply to hypergraphs. Our key assumption was that the connectivity structure arises via a geometric 
mechanism---hyperedges are formed by closeness to
(unknown) hyperedge centres. This allowed us to
define a goodness-of-fit for an embedding. We began by considering a spectral embedding based on the graph Laplacian associated with the bipartite graph representation of the hypergraph. We then showed how the embedding can be improved by considering a smoothed version of the goodness-of-fit and allowing the bipartite weights and the connectivity radius to be 
iteratively refined using gradient descent. 

We derived Gradient-Descent Spectral Embedding (GDSE, Algorithm~\ref{a:gdse}), which computes a new spectral embedding at each iteration. We also proposed the Gradient-Descent Embedding (GDE) approach (Algorithm~\ref{a:gde}), which relaxes the requirement that the final embedding should be spectral and is practical for large, sparse hypergraphs. We then showed that GDSE and GDE are effective at recovering the geometric make-up of synthetically generated data. On real hypergraphs, in addition to quantifying the level of geometric structure, the algorithms were shown to be useful for the downstream tasks
of predicting missing or spurious hyperedge memberships and community detection. 

\looseness-1
We note that the loss function introduced here can be identified with a likelihood function for the geometric random hypergraph model proposed in 
\cite{RandomGeometricHypergraph2025}. It is feasible to introduce other random hypergraph models, for example 
with hyperedge membership arising independently
at random, or driven by a periodic distance measure, and thereby to use likelihood ratios to compare modelling hypotheses, as in~\cite{HypergraphEmbeddingNature,GrindrodPeter2010Pr}.

The choice of embedding dimension is also a key ingredient. In this work, we used the heuristic approach of examining the loss function over a range of possibilities. A more systematic alternative would be to extend the two-nearest neighbour dimension estimation algorithm from 
\cite{Facco17}, which has proved useful in the graph setting \cite{GHKstruggle24}.

\section*{Data statement}
The relevant Python code used for the experiments presented here is available in a GitHub repository \texttt{\href{https://github.com/francesco-zigliotto/hypergraph-embedding/blob/main/hypergraph_embedding.ipynb}{francesco-zigliotto/hypergraph-embedding}}.

\section*{Acknowledgements}
FZ acknowledges funding from INdAM-GNCS (Project ``Metodi basati su matrici e tensori strutturati per problemi di algebra lineare di grandi dimensioni'') and from MUR (Ministero dell'Universit\`a e della Ricerca) through the PRIN Project 20227PCCKZ (``Low Rank Structures and Numerical Methods in Matrix and Tensor Computations and their Applications'').
DJH was supported by a Fellowship from the Leverhulme Trust
 and by the Advanced Grant ``Numerical Analysis for Stable AI''
  101198795 from the European Research Council.

\section*{\refname}
\hfuzz2pt
\printbibliography[heading=none]

\end{document}